\documentclass[
aps,
pra,
twocolumn,
10pt,
notitlepage,
]{revtex4-2}

\usepackage[T1]{fontenc}
\usepackage[english]{babel}

\usepackage{lmodern}
\usepackage{microtype}
\usepackage{float}

\usepackage{mathtools}
\usepackage{amssymb}
\usepackage{amsthm}

\usepackage{graphicx}

\usepackage{tikz}
\usetikzlibrary{
decorations.pathreplacing,
arrows.meta,
positioning,
calc,
calligraphy,
matrix,
fit
}

\tikzset{
    causal graph/.style={
        >=Stealth,
        scale=0.9,
        transform shape,
        node distance=1.7cm and 2.5cm,
        every node/.style={
            draw,
            circle,
            minimum size=0.75cm,
            inner sep=0pt,
            align=center
        }
    }
}

\newtheorem{theorem}{Theorem}[section]
\newtheorem{definition}{Definition}[section]

\usepackage[
colorlinks=true,
linkcolor=black,
citecolor=black,
urlcolor=black
]{hyperref}

\usepackage[nameinlink]{cleveref}

\begin{document}

\title{Gisin's Argument and the Limits of Causal Explanations in Relativistic Spacetime}
\author{Felix J. Rutzinger}
\noaffiliation

\begin{abstract}
Gisin has provided an argument for the conclusion that no covariant nonlocal model can reproduce the operational statistics observed in Bell experiments. Gisin's argument applies only to deterministic models and we argue that proposed generalizations of the argument beyond determinism based solely on statistical notions are unconvincing, as they lack the resources to formulate a satisfactory notion of no-retrocausality. We introduce the framework of frame-indexed causal models in order to extend Gisin's argument to the genuinely probabilistic case. Frame-indexed causal models associate potentially different causal models with  different reference frames while requiring corresponding operational variables to agree runwise. Within this framework, we establish two no-go results based on Gisin's argument. First, we show that no empirically adequate frame-indexed causal model can jointly satisfy Independent Settings, No-Retrocausality, and Causal Lorentz Invariance. Second, it is established that this result continues to hold when Causal Lorentz Invariance is weakened to Lorentz Invariance of Causal Connections, which requires only the skeleton of the causal structure to remain invariant, thereby allowing the direction of causal influence to depend on the choice of reference frame. These results are discussed in the context of relativistic interpretations of quantum theory and nonclassical causal inference.
\end{abstract}
\maketitle

\tableofcontents

\section{Introduction}\label{sec:I}
The apparent conflict between nonlocality and special relativity has long been a central challenge facing realist approaches to quantum phenomena. This was recognized early on by John Bell, who, in an interview given shortly before his untimely death, remarked the following:
\begin{quote}
    So one of my missions in life is to get people to see that if they want to talk about the problems of quantum mechanics---the real problems of quantum mechanics---they must be talking about Lorentz invariance. \cite{Bell_1988}
\end{quote}
The difficulty of this ``real problem of quantum mechanics'' is reflected by the fact that arguably the first genuinely relativistic realist theory of quantum phenomena, Roderich Tumulka's relativistic GRW model with flash ontology (rGRWf), was not proposed until 2006 \cite{Tumulka_2006}.

Although the conflict between nonlocality and relativity is a pressing foundational problem, the exact nature of this tension is far from clear \cite{ Ghirardi_2010, Maudlin_2011, Duerr_2014, Myrvold_2021,  Allori_2022}. The question whether relativity imposes principled, nontrivial constraints on nonlocal explanations of Bell correlations remains largely open.

A possible answer to this question has been offered by Nicholas Gisin \cite{Gisin_2010, Gisin_2011}. In a 2011 paper, he formulated an impossibility argument which aims to show that ``there is no covariant nonlocal model of quantum correlations'' \cite{Gisin_2011}. Despite its strength, this claim has received limited attention in the literature, with some notable exceptions \cite{Levy_2020, Scarani_2019}.

In this paper, we analyze Gisin's argument and clarify the assumptions involved. We argue that his argument, as formulated in \cite{Gisin_2011}, only applies to deterministic models of Bell experiments. Two proposed extensions of the argument beyond determinism are discussed and found to be unconvincing. The framework of frame-indexed causal models is introduced and shown to provide the adequate resources for formalizing the ideas underlying Gisin's argument. Using frame-indexed causal models, we derive two no-go theorems, thereby extending Gisin's original argument to the genuinely probabilistic case. The first and most straightforward generalization of Gisin's argument via causal models requires an additional invariance assumption on causal structure, not present in the original argument. This condition, called Causal Lorentz Invariance, states that the causal structure is independent of the choice of reference frame. This result is then strengthened by replacing Causal Lorentz Invariance with a strictly weaker invariance condition referred to as Lorentz Invariance of Causal Connections, which states that only the skeleton of the causal structure must be frame-independent. Finally, these no-go theorems are discussed in the context of nonclassical causal inference and realist theories of quantum phenomena in relativistic spacetime, namely relativistic versions of Bohmian mechanics and relativistic collapse theories. In the Appendix, an additional no-go result is derived which does not make use of an invariance assumption on causal structure. Instead it primarily uses a principle called Independent Explanation, which is motivated by the idea that if the same experimental situation is described by differing causal structures in different frames, the experimental outcomes are generated by independent mechanisms. The proof proceeds by showing that Independent Explanation together with additional assumptions forces determinism, thereby allowing Gisin's original argument to go through.

\section{Gisin's Argument}\label{sec:II}
Gisin's argument concerns a bipartite Bell experiment viewed from two different reference frames. In a bipartite Bell experiment, the two experimenters, Alice and Bob, each perform their respective experiments at spacelike separation. Alice chooses her setting $X$ and obtains her outcome $A$, while Bob chooses his setting $Y$ and obtains his outcome $B$. The possible values of the measurement outcomes and settings are denoted by $a \in \{0, \dots, |A|-1\}$, and analogously for the other operational variables $B$, $X$ and $Y$.

In the reference frame $F$, Alice performs her experiment before Bob, meaning that she observes her outcome $A$ even before Bob chooses his setting $Y$. In $F'$, the temporal ordering is reversed, such that Bob performs his experiment before Alice. It is always possible to find two such frames as the experiments are performed at spacelike separation. 

Gisin assumes that the outcomes of the Bell experiment $A$ and $B$ are functions of the settings $X$ and $Y$ and $\Lambda$. 
\begin{align}
    a &= f(x, y, \lambda)\,\\
    b &= g(x, y, \lambda)\,.
\end{align}
The hidden variable $\Lambda$ denotes possible additional influences on the outcomes of the Bell experiment. This assumption restricts the argument to deterministic models of the Bell experiment, since the values of $X$, $Y$ and $\Lambda$ uniquely determine the value of the outcomes $A$ and $B$.

Furthermore, the experimenters are assumed to be free in their choice of measurement setting, meaning that the settings do not depend on the other variables relevant to the Bell experiment, or on variables which share a common cause with the relevant variables.  

Gisin assumes that in a reference frame where the value of a variable is determined before another variable, the first variable may not functionally depend on the second. Hence, in $F$ we have
\begin{align}
    a &= f(x, \lambda)\,,\\
    b &= g(x, y, \lambda)\,,
\end{align}
as in this frame $A$ is determined before $B$ \footnote{Notice that any nonlocal dependence between the outcomes can be absorbed into the response functions without loss of generality: $b = h(a, x, y, \lambda) = h(f(x, \lambda), x, y, \lambda) \equiv g(x, y, \lambda)$.}. Analogously, in $F'$, where the temporal order is reversed, we find that
\begin{align}
    a' &= f'(x', y', \lambda')\,,\\
    b' &= g'(y', \lambda')\,,
\end{align}
where the primed variables denote the variables as seen from $F'$. This assumption is a deterministic formulation of no-retrocausality, the general principle that the past must not depend on the future.

Gisin then argues that in a ``covariant model'' all the primed quantities must be equal to their unprimed counterparts. This allows us to deduce the existence of the following model of the Bell experiment
\begin{align}
    a &= f(x, \lambda)\,,\\
    b &= g'(y, \lambda)\,.
\end{align}
Hence, we can express the operational statistics $p(a, b \mid x, y)$ as 
\begin{equation}
    p(a, b \mid x, y) =  \int \mathrm{d}\lambda \,\delta_{a, f(x, \lambda)}\,\delta_{b, g'(y, \lambda)}\,p(\lambda)\,,
\end{equation}
which is Bell-local. Here, we have made use of the assumption that the experimenters are free to choose their measurement settings, which we take to imply \emph{Measurement Independence} (MI):
\begin{equation}
    p(\lambda \mid x, y) = p(\lambda)\,.
\end{equation} 
For binary-valued variables $|A| = |B| = |X| = |Y| = 2$, this implies the Clauser-Horne-Shimony-Holt (CHSH) inequality \cite{Clauser_1969}
\begin{equation}
    |E_{00} + E_{01} + E_{10} - E_{11}| \leq 2\,,
\end{equation}
where $E_{xy} = \sum_{a,b} (-1)^{a + b}\,p(a, b \mid x, y)$, which is violated experimentally \cite{Giustina_2015} and can be violated by quantum theory up to the Tsirelson bound of $2\sqrt{2}$ \cite{Tsirelson_1980}. Hence, any deterministic model satisfying the conditions described above must be empirically inadequate. 

It should be noted that Gisin's assumption of ``covariance'' is not a substantive assumption when applied to operationally accessible variables. Assume that both Alice and Bob record their respective measurement settings and outcomes physically, for example by writing them on lists. Such a physical record associated with a measurement setting or outcome is not altered by a change of reference frame. Hence, for each run of the experiment, all values of the operational variables $A$, $B$, $X$, and $Y$ must be independent of the choice of reference frame.  

For the hidden variable, it is not necessary to assume that $\lambda = \lambda'$. Instead it suffices to assume a cross-frame version of measurement independence
\begin{equation}
    p(\lambda, \lambda' \mid x, y) = p(\lambda, \lambda')\,,
\end{equation} 
in order to infer Bell locality:
\begin{align}
    p(a, b \mid x, y) &= 
    \int \mathrm{d}\lambda\,\mathrm{d}\lambda' \,\delta_{a, f(x, \lambda)}\,\delta_{b, g'(y, \lambda')}\,p(\lambda, \lambda' \mid x, y)\\
    &= \int \mathrm{d}\lambda\,\mathrm{d}\lambda' \,\delta_{a, f(x, \lambda)}\,\delta_{b, g'(y, \lambda')}\,p(\lambda, \lambda')\,.
\end{align}
Gisin concludes from this argument that covariant and nonlocal explanations of Bell experiments are impossible \cite{Gisin_2011}. This rather strong conclusion does not follow from the argument presented above. We have seen that Gisin's reasoning relies on three substantial assumptions: first, that the experimenters are free to choose their measurement settings independently of the other variables relevant to the Bell experiment; second, that variables cannot retrocausally depend on variables in their future; third, and most importantly for what follows, that the Bell experiment is correctly described by a deterministic model.

\section{Probabilistic Generalizations}\label{sec:III}
\subsection{Deterministic Simulation}\label{subsec:IIIA}
The argument, as formulated above, is restricted to the deterministic case as it assumes that the outcomes are, for all runs of the experiments, functions of the settings and the hidden variable. Gisin suggests extending the argument by  simply adding additional random variables \cite{Gisin_2011}. In \cite{Esfeld_2013}, Esfeld and Gisin rightly observe that every probabilistic behavior can be deterministically simulated in this way.

This is exactly the same strategy Conway and Kochen propose for converting any stochastic model into a deterministic one, in the context of their ``Free Will Theorem'':
\begin{quote}
    [L]et the stochastic element [\dots] be a sequence of random numbers [\dots]. Although these might only be generated as needed, it will plainly make no difference to let them be given in advance. \cite{Conway_2006} 
\end{quote}
Goldstein et al. argued that this is not enough. One further needs to check whether the assumptions of the relevant no-go argument are true, or at least convincing, in the genuinely stochastic case, something which needs to be established on its own merits \cite{Goldstein_2010}. 

As an example to illustrate this point, they cite the case of Outcome versus Parameter Independence \cite{Goldstein_2010}. Experimental violations of the CHSH inequality show that any empirically adequate behavior describing a Bell experiment must violate Bell locality. Bell locality is jointly implied by \emph{Outcome Independence} (OI)
\begin{equation}
    p(a, b \mid x, y, \lambda) = p(a \mid x, y, \lambda)\,p(b \mid x, y, \lambda)\,
\end{equation}
\emph{Parameter Independence} (PI)
\begin{align}
    p(a \mid x, y, \lambda) &= p(a \mid x, \lambda)\,,\\
    p(b \mid x, y, \lambda) &= p(b \mid y, \lambda)\,,
\end{align}
and \emph{Measurement Independence} (MI)
\begin{equation}
    p(\lambda \mid x, y) = p(\lambda)\,.
\end{equation}
Under the assumption of MI, the question arises whether OI or PI or both fail. This is a nontrivial issue which has led to a sizable literature \cite{Jarrett_1984, Chaves_2015, Ringbauer_2016, Vieira_2025}. 

Since deterministic models satisfy OI trivially, any empirically adequate deterministic model of a Bell experiment that respects MI must violate PI. Hence, by reasoning analogous to that of Esfeld and Gisin, one would have to conclude that probabilistic models, given that they are empirically adequate and satisfy MI, must violate PI as well. However, this is not generally true for probabilistic theories, as there are realist interpretations of quantum theory, for example the objective collapse models such as rGRWf, which respect PI while violating OI \cite{Tumulka_2009}. 

In the following we will consider possible probabilistic formulations of Gisin's argument and argue that the assumption excluding retrocausal influences cannot be satisfactorily formulated using only statistical notions.

\subsection{Ontological Models}\label{subsec:IIIB}
To generalize Gisin's argument beyond the deterministic case one needs to express the relevant assumptions in a genuinely probabilistic framework. This approach has been proposed by Levy and Hemmo \cite{Levy_2020, Levy_2025}. They have tried to recast Gisin's argument using their own probabilistic framework. In the following, the well-known framework of ontological models is used instead \cite{Harrigan_2010, Leifer_2014}.  

It is possible to generalize Gisin's argument in this way, as has been done in the Appendix, but the result does not amount to an interesting no-go theorem. The difficulty lies in the fact that it is not clear how the condition of no-retrocausality can be formulated in a reasonable way using only statistical notions. The idea behind any principle of no-retrocausality seems to be that it should prohibit influence from the future to the past. A naive way of trying to express this idea using ontological models is the following: 
\begin{definition}[Probabilistic No-Retrocausality]
    All variables are conditionally independent of variables in their future, given a full specification of their past.
\end{definition}
Using \emph{Probabilistic No-Retrocausality} (PNRC), the factorization $p(a, b \mid x, y, \lambda) = p(a \mid x, \lambda)\,p(b \mid y, \lambda)$ can be derived by again considering reference frames with different time orders. In $F$, PNRC implies that
\begin{equation}
    p(a \mid b, x, y, \lambda) = p(a \mid x, \lambda)\,,
\end{equation}
as $B$ and $Y$ are in the future of $A$ while $A$'s past is fully specified by $X$ and $\Lambda$. In $F'$
\begin{equation}
    p(b \mid x, y, \lambda) = p(b \mid y, \lambda)\,,
\end{equation}
must obtain as $X$ is in the future of $B$, while $B$'s past is fully specified by $Y$ and $\Lambda$. If one further assumes that statistical independence relations are frame-independent together with MI, then the resulting operational statistics $p(a, b \mid x, y)$ must be Bell-local. 

To see why PNRC is unconvincing consider a simple Markov chain consisting of three random variables $X_0, X_1, X_2$, their index corresponding to their temporal order. These three random variables describe the state of some physical system and jointly model its time evolution.
\begin{figure}[H]
\begin{center}
\begin{tikzpicture}[causal graph]

\node (x0) {$X_0$};
\node[right=of x0] (x1) {$X_1$};
\node[right=of x1] (x2) {$X_2$};

\draw[->] (x0) -- (x1);
\draw[->] (x1) -- (x2);

\end{tikzpicture}
\end{center}
\caption{A simple Markov chain representing the time evolution of a physical system.}
\end{figure}
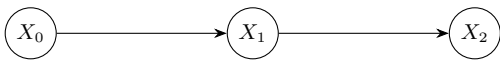
PNRC implies that $X_1$ is independent of $X_2$ given $X_0$
\begin{equation}
    p(x_1, x_2 \mid x_0) = p(x_1 \mid x_0)\,p(x_2 \mid x_0)\,,
\end{equation}
while the Markov condition implies 
\begin{equation}
    p(x_1, x_2 \mid x_0) = p(x_1 \mid x_0)\,p(x_2 \mid x_1)\,.
\end{equation}
From this it follows that 
\begin{equation}\label{eq:1}
    p(x_2 \mid x_0) = p(x_2 \mid x_1)\,,
\end{equation}
meaning that $X_1$ contains no information about $X_2$ beyond that already contained in $X_0$. But this is exactly the kind of condition which one expects to fail in physical theories which exhibit genuinely stochastic time evolution. In these theories, new information, which then influences the evolution of the system, may enter with time. Hence, this version of no-retrocausality is generally incompatible with genuinely stochastic physical theories whose dynamics involve the creation of new information over time \footnote{There may exist stochastic physical theories which respect equation \eqref{eq:1} but they are generally rather unconvincing. One could for example have a theory in which the physical state at all times is independent of the physical state at all other times.}. Hence, any generalization of Gisin's argument beyond deterministic theories which makes use of this version of no-retrocausality is uninteresting. 

One may try to avoid these problems by using a more sophisticated version of PNRC while still remaining inside the ontological models framework. Such a condition has been offered, although in a context not directly related to Gisin's argument, by Leifer and Pusey \cite{Leifer_2017}. They distinguish between input ($X$, $Y$, $\Lambda$) and noninput variables ($A$, $B$) to avoid the problems discussed above:
\begin{quote}
    All noninput variables are conditionally independent of input variables in their future, given a full specification of their past.
\end{quote}
The problem with this notion for present purposes is that it is too weak to derive Bell locality. It allows us to derive PI but not OI because neither $A$ nor $B$ is an input variable. Hence, this sophisticated probabilistic notion of no-retrocausality does not imply the factorization
\begin{equation}
    p(a, b \mid x, y, \lambda) = p(a \mid x, \lambda)\,p(b \mid y, \lambda)\,,
\end{equation}
which, together with MI, would imply Bell locality. 

In the rest of the paper we show that these issues can be overcome by making use of methods borrowed from the field of causal inference. Frame-indexed causal models allow no-retrocausality to be expressed as a constraint on causal structure and not on statistical independence relations, thereby avoiding the difficulties noted above. 

\section{Causal Model Formulations}\label{sec:IV}
\subsection{An Introduction to Frame-Indexed Causal Models}\label{subsec:IVA}
In the following we consider causal explanations of the same physical experiment relative to different reference frames. In order to give a precise account of how potentially frame-dependent causal explanations of the same experiment relate to each other, the notion of a frame-indexed causal model is introduced. Before doing so, relevant notions from causal inference are briefly reviewed. Here, we will mostly follow the presentation in \cite{Wood_2015}.

A \emph{causal structure} $G$ is a directed acyclic graph (DAG) $(\mathbf{V}, E)$ where $\mathbf{V}$ is a finite set of random variables $\{X_1, \dots, X_n\}$ and $E \subseteq \mathbf{V} \times \mathbf{V}$ is a set of ordered pairs of distinct variables $(X_i, X_j)$, specifying that $X_i$ is a direct cause of $X_j$ in $G$. The graph $G$ is directed as each edge has an orientation and it is acyclic because it is assumed to contain no directed paths that begin and end at the same vertex.

If $X_j$ is a direct cause of $X_i$ then $X_j$ is called a \emph{parent} of $X_i$. The set of all parents of $X_i$ is denoted by $\mathrm{Pa}(X_i)$. An important distinction is that between variables which are operationally accessible and those which are latent. In the context of the Bell scenario, the former will as usual be denoted by Latin letters such as $A$, $B$, $X$, and $Y$, whereas the latter will be denoted by Greek letters such as $\Lambda$. Furthermore, the measurement outcomes are denoted by letters from the beginning of the alphabet, such as $A$ and $B$, while the measurement settings are denoted by letters from the end of the alphabet, such as $X$ and $Y$.

A \emph{causal model} $C$ is a pair $(G,P)$, where $G$ is a causal structure and $P$ is a joint probability distribution over the variables in $\mathbf{V}$ which is Markov relative to $G$, meaning that
\begin{equation}
    P(X_1,\ldots,X_n) = \prod_{i=1}^{n} P(X_i\mid\mathrm{Pa}(X_i))\,.
\end{equation}
We use uppercase $P$ to denote probability distributions over random variables, and lowercase $p$ for the probabilities of particular values of random variables, or probability densities in the case of continuous variables. Thus,
\begin{equation}
    p(x_1,\ldots,x_n)
    =
    P(X_1=x_1,\ldots,X_n=x_n)\,.
\end{equation}
In order to allow for causal explanations that may depend on the choice of reference frame, we employ the framework of \emph{frame-indexed causal models}, formally developed in Appendix. The basic idea is to associate with each inertial reference frame $F$ a causal model
\begin{equation}
    C^F = (G^F, P^F)\,,
\end{equation}
defined over a corresponding set of frame-indexed random variables $\mathbf{V}^F$. Variables with the same index in different frames are localized in the same spacetime region $R_i$ of Minkowski spacetime. This allows the temporal relations between variables to be evaluated relative to each reference frame and thereby enables us to apply the principle of no-retrocausality. The idea of embedding random variables in spacetime has been explored more in depth by Vilasini and Colbeck \cite{Vilasini_2022}. Furthermore, the distinction between nonlatent and latent variables is assumed to be independent of reference frame. 

The set of all frame-indexed variables
\begin{equation}
    \mathbf{V} = \bigcup_{j = 1}^{m}\mathbf{V}^{F_j}\,,
\end{equation}
is assumed to admit a joint probability distribution $P(\mathbf{V})$ across the different causal models. For each frame $F$, the distribution $P^F$ is the marginal of $P$ over the variables in $\mathbf{V}^F$. This allows probabilistic relations between different frame-indexed descriptions to be formulated without identifying the corresponding random variables from the outset. Since the models describe the same physical experiment, corresponding operational variables are required to agree in each individual run. Thus, for any two frames $F$ and $F'$,
\begin{equation}
    P\left(X_i^F = X_i^{F'}\right) = 1\,,
\end{equation}
for every nonlatent variable $X_i$. We refer to this condition as \emph{Runwise Agreement}. It is stronger than observational equivalence, since it requires agreement of the individual outcomes rather than merely agreement of their probability distributions.

Crucially, the framework itself imposes no constraint on how the causal structure $G^F$ depends on the choice of reference frame. Different assumptions concerning this dependence will be considered in the following.

\subsection{Causal Principles}\label{subsec:IVB}
The operationally accessible variables $A, B, X, Y$ of the Bell experiment, together with the hidden variable $\Lambda$, are associated with nodes of a causal structure. The hidden variables are assumed to include all latent causal influences on the operational variables of the Bell experiment. 

The following arguments proceed by the usual strategy: the assumptions restrict the set of causal structures compatible with the Bell experiment. By the Markov condition, this restricts the set of compatible joint probability distributions to distributions that are empirically inadequate. Consequently, at least one of the causal assumptions must fail.

In the case of the unsuccessful ontological models versions of Gisin's argument, one needs to assume an invariance condition on statistical independence relations. Similarly for causal models, the most straightforward generalizations of Gisin's argument make use of an invariance condition on causal structure. In the following, this assumption will be successively relaxed. Therefore the relevant invariance conditions will be discussed jointly with the corresponding no-go arguments. 

All formulations of the argument rely on the assumption of \emph{Independent Settings} (IS), which enforces the freedom of experimenters to choose their settings independently of the experiment they are performing: 
\begin{definition}[Independent Settings]
The measurement settings have no relevant causes.
\end{definition}
This means that the measurement settings $X$ and $Y$ do not have causes among, nor share a common cause with, any of the other variables relevant to the Bell experiment in all frames  \cite{Ying_2024}.

The condition of \emph{No-Retrocausality} (NRC), which excludes causal influence from the future to the past, is defined as follows:
\begin{definition}[No-Retrocausality]
    No variable can be caused by a variable in its future.
\end{definition}
This assumption is taken to apply to each reference frame individually. This means that for any reference frame $F$, if $X_j^F$ is in the future of $X_i^F$ in $F$, then $X_j^F$ cannot be a parent of $X_i^F$. As with IS, NRC is assumed throughout.

\subsection{Causal Lorentz Invariance}\label{subsec:IVC}
The most straightforward way of formulating Gisin's argument using frame-indexed causal models is to assume, in conjunction with IS and NRC, that the causal structure is independent of the choice of reference frame. This condition will be referred to as \emph{Causal Lorentz Invariance} (CLI):
\begin{definition}[Causal Lorentz Invariance]
    The causal structure is identical in all reference frames.
\end{definition}
This principle gives a transformation rule for causal structure lacking from the framework of frame-indexed causal models introduced above. It states that the causal structure transforms trivially under a change of Lorentz frames. Using CLI, the following no-go result can be established:
\begin{theorem}\label{theorem:1}
    There is no empirically adequate, frame-indexed causal model of the Bell scenario that satisfies:
    \begin{itemize}
        \item[(i)] Independent Settings,
        \item[(ii)] No-Retrocausality,
        \item[(iii)] Causal Lorentz Invariance.
    \end{itemize}
\end{theorem}
\begin{proof}
The Bell experiment is assumed to be localized in spacetime as discussed in Section~\ref{sec:II}.

IS implies that the measurement settings $X$ and $Y$ have no relevant causal parents. This implies that the latent variable $\Lambda$ can at most act as a common cause of the outcomes $A$ and $B$. Furthermore, IS excludes causal influence from the outcomes to the settings and direct causal relations between the settings. This leaves us with the following class of allowed causal structures, independent of the choice of reference frame:
\begin{figure}[H]
\centering
\begin{tikzpicture}[causal graph]

\node (a) at (0,0) {$A$};
\node (b) at (2.5,0) {$B$};
\node (x) at (0,-1.7) {$X$};
\node (y) at (2.5,-1.7) {$Y$};
\node (lambda) at (1.25,-3.4) {$\Lambda$};

\draw[-] (a) -- (b);
\draw[->] (x) -- (a);
\draw[->] (y) -- (b);
\draw[->] (x) -- (b);
\draw[->] (y) -- (a);
\draw[->] (lambda) -- (a);
\draw[->] (lambda) -- (b);

\end{tikzpicture}
\caption{The class of causal structures underlying the Bell experiment compatible with IS.}
\end{figure}
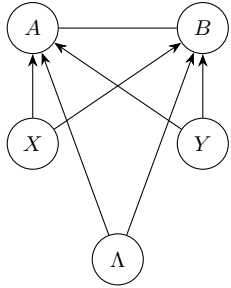
Here, the straight line between $A$ and $B$ represents the fact that both $A \rightarrow B$ and $B \rightarrow A$ are compatible with IS.

To apply NRC, we must consider the causal structure from different reference frames. As before, in reference frame $F$, Alice performs her experiment before Bob, while in frame $F'$, the temporal order is reversed. In $F$, NRC rules out $B$ and $Y$ as causes of $A$, since $B$ is observed and $Y$ is chosen only after $A$ has been determined, whereas in $F'$, $A$ and $X$ cannot cause $B$ for analogous reasons. This leaves us with the following allowed causal structures in $F$ and $F'$, respectively:
\begin{figure}[H]
\centering
\begin{minipage}{0.48\columnwidth}
\centering
$F$\\[0.8ex]
\begin{tikzpicture}[causal graph]
 
  \node[draw, circle, minimum size=0.75cm, inner sep=0pt] (a) at (0,0) {$A$};
  \node[draw, circle, minimum size=0.75cm, inner sep=0pt] (b) at (2.5,0) {$B$};
  \node[draw, circle, minimum size=0.75cm, inner sep=0pt] (x) at (0,-1.7) {$X$};
  \node[draw, circle, minimum size=0.75cm, inner sep=0pt] (y) at (2.5,-1.7) {$Y$};
  \node[draw, circle, minimum size=0.75cm, inner sep=0pt] (lambda) at (1.25,-3.4) {$\Lambda$};

  \draw[->] (a) -- (b);
  \draw[->] (x) -- (a);
  \draw[->] (x) -- (b);
  \draw[->] (y) -- (b);
  \draw[->] (lambda) -- (a);
  \draw[->] (lambda) -- (b);
  
\end{tikzpicture}
\end{minipage}
\hfill
\begin{minipage}{0.48\columnwidth}
\centering
$F'$\\[0.8ex]
\begin{tikzpicture}[causal graph]

  \node[draw, circle, minimum size=0.75cm, inner sep=0pt] (a) at (0,0) {$A'$};
  \node[draw, circle, minimum size=0.75cm, inner sep=0pt] (b) at (2.5,0) {$B'$};
  \node[draw, circle, minimum size=0.75cm, inner sep=0pt] (x) at (0,-1.7) {$X'$};
  \node[draw, circle, minimum size=0.75cm, inner sep=0pt] (y) at (2.5,-1.7) {$Y'$};
  \node[draw, circle, minimum size=0.75cm, inner sep=0pt] (lambda) at (1.25,-3.4) {$\Lambda'$};

  \draw[->] (b) -- (a);
  \draw[->] (x) -- (a);
  \draw[->] (y) -- (b);
  \draw[->] (y) -- (a);
  \draw[->] (lambda) -- (a);
  \draw[->] (lambda) -- (b);

\end{tikzpicture}
\end{minipage}
\caption{The maximally allowed causal structure underlying the Bell experiment compatible with IS and NRC in $F$ and $F'$.}
\label{figure:1}
\end{figure}
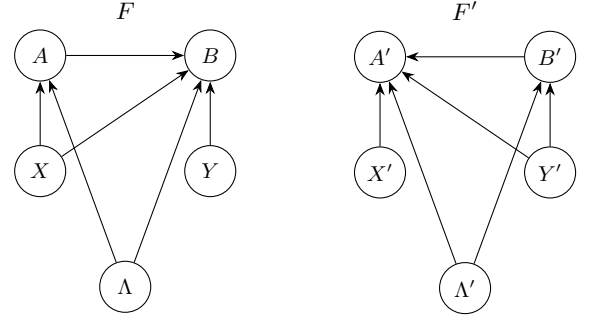
CLI requires the causal structure to be identical in all reference frames. Therefore, if a causal influence is ruled out in one frame, it is ruled out in all frames. The only allowed arrows are consequently those present in both frames, which rules out any direct causal influence between the parties and leaves the Bell-local causal structure shown below:
\begin{figure}[H]
\centering
\begin{tikzpicture}[causal graph]

\node (a) at (0,0) {$A$};
\node (b) at (2.5,0) {$B$};
\node (x) at (0,-1.7) {$X$};
\node (y) at (2.5,-1.7) {$Y$};
\node (lambda) at (1.25,-3.4) {$\Lambda$};

\draw[->] (x) -- (a);
\draw[->] (y) -- (b);
\draw[->] (lambda) -- (a);
\draw[->] (lambda) -- (b);

\end{tikzpicture}
\caption{The maximally allowed causal structure underlying the Bell experiment compatible with IS, NRC, and CLI.}
\end{figure}
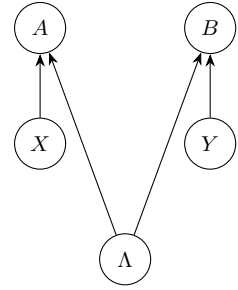
Applying the Markov condition, we obtain that the operational statistics must obey Bell locality:
\begin{equation}
    p(a, b \mid x, y) = \int \mathrm{d}\lambda\, p(a \mid x, \lambda)\,p(b \mid y, \lambda)\,p(\lambda)\,,
\end{equation}
which implies the CHSH inequality
\begin{equation}
    |E_{00} + E_{01} + E_{10} - E_{11}| \leq 2\,,
\end{equation}
for binary-valued nonlatent variables $|A| = |B| = |X| = |Y| = 2$. Since the CHSH inequality is experimentally violated, any causal model jointly satisfying IS, NRC, and CLI is unable to reproduce the operational statistics and is therefore empirically inadequate. 
\end{proof}

\subsection{Lorentz Invariance of Causal Connections}\label{subsec:IVD}
A natural relaxation of Causal Lorentz Invariance is obtained by only requiring \emph{Lorentz Invariance of Causal Connections} (LICC). It can be stated as follows:
\begin{definition}[Lorentz Invariance of Causal Connections]
    The skeleton of the causal structure is identical in all reference frames.
\end{definition}
LICC allows the direction of causal arrows to depend on the choice of reference frame while enforcing that the skeleton of the causal structure, i.e., the undirected graph, remains invariant. 

One motivation for LICC stems from the idea that the usual frame-independent notion of causality may no longer generally apply in genuinely relativistic and stochastic theories of quantum phenomena, such as rGRWf. Tumulka, for example, has commented on this issue in a discussion of the ``Free Will Theorem'' \footnote{As discussed in the Conclusion, Tumulka at other times suggests there is no need for causal notions at the level of fundamental physics. Hence, his position may not amount to a rejection of CLI in favor of LICC but rather to a wholesale rejection of the applicability of causal notions to the Bell scenario.}:
\begin{quote}
    Conway and Kochen have even introduced a word for the frame-dependence of influences: effective causality. [\dots] Thus, the ``causality principle'' means that one event cannot influence an earlier one. [\dots] A theory is effectively causal if in each frame the causality principle holds. This is the case, for example, in rGRWf, because whether $f_A$ influences $f_B$ [here $f_A$ or $f_B$ denote flashes] or vice versa is frame-dependent. That is, rGRWf is effectively causal despite the nonlocality. \cite{Tumulka_2007}
\end{quote}
While LICC is strictly weaker than CLI, it, in conjunction with NRC, still provides a transformation rule for causal structure for almost all inertial frames. LICC enforces that the skeleton of the causal structure remains invariant. NRC fixes the direction of influence in every frame, except for those frames where causally connected variables are determined simultaneously. Furthermore, we can use LICC to establish the following strictly stronger no-go theorem:
\begin{theorem}\label{theorem:2}
 There is no empirically adequate, frame-indexed causal model of the Bell scenario that satisfies:
    \begin{itemize}
        \item[(i)] Independent Settings,
        \item[(ii)] No-Retrocausality,
        \item[(iii)] Lorentz Invariance of Causal Connections.
    \end{itemize}
\end{theorem}
\begin{proof}
As in the proof of Theorem~\ref{theorem:1}, we constrain the allowed causal structures using IS and NRC, by considering them in $F$ and $F'$, respectively. Again, this leads to the following pair of causal structures:
\begin{figure}[H]
\centering
\begin{minipage}{0.48\columnwidth}
\centering
$F$\\[0.8ex]
\begin{tikzpicture}[causal graph]
 
  \node[draw, circle, minimum size=0.75cm, inner sep=0pt] (a) at (0,0) {$A$};
  \node[draw, circle, minimum size=0.75cm, inner sep=0pt] (b) at (2.5,0) {$B$};
  \node[draw, circle, minimum size=0.75cm, inner sep=0pt] (x) at (0,-1.7) {$X$};
  \node[draw, circle, minimum size=0.75cm, inner sep=0pt] (y) at (2.5,-1.7) {$Y$};
  \node[draw, circle, minimum size=0.75cm, inner sep=0pt] (lambda) at (1.25,-3.4) {$\Lambda$};

  \draw[->] (a) -- (b);
  \draw[->] (x) -- (a);
  \draw[->] (x) -- (b);
  \draw[->] (y) -- (b);
  \draw[->] (lambda) -- (a);
  \draw[->] (lambda) -- (b);
  
\end{tikzpicture}
\end{minipage}
\hfill
\begin{minipage}{0.48\columnwidth}
\centering
$F'$\\[0.8ex]
\begin{tikzpicture}[causal graph]

  \node[draw, circle, minimum size=0.75cm, inner sep=0pt] (a) at (0,0) {$A'$};
  \node[draw, circle, minimum size=0.75cm, inner sep=0pt] (b) at (2.5,0) {$B'$};
  \node[draw, circle, minimum size=0.75cm, inner sep=0pt] (x) at (0,-1.7) {$X'$};
  \node[draw, circle, minimum size=0.75cm, inner sep=0pt] (y) at (2.5,-1.7) {$Y'$};
  \node[draw, circle, minimum size=0.75cm, inner sep=0pt] (lambda) at (1.25,-3.4) {$\Lambda'$};

  \draw[->] (b) -- (a);
  \draw[->] (x) -- (a);
  \draw[->] (y) -- (b);
  \draw[->] (y) -- (a);
  \draw[->] (lambda) -- (a);
  \draw[->] (lambda) -- (b);

\end{tikzpicture}
\end{minipage}
\caption{The maximally allowed causal structure underlying the Bell experiment compatible with IS and NRC in $F$ and $F'$.}
\end{figure}
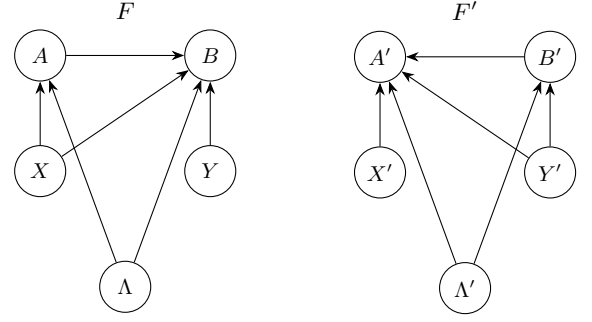
LICC states that the skeleton of a causal structure must be invariant under changes of reference frame. This implies that if a random variable $X_i$ is a cause of another variable $X_j$ in some frame then in all frames, either $X_i$ is a cause of $X_j$ or $X_j$ is a cause of $X_i$. 

In $F$, $X$ may be a cause of $B$ while in $F'$ neither $X'$ is a cause of $B$' nor $B'$ of $X'$. Hence, $X$ cannot be a cause of $B$ in any frame. By an analogous argument, $Y$ cannot be a cause of $A$ in any frame. LICC does not rule out causal relations between $A$ and $B$, as $A$ may be a cause of $B$ in $F$, while the reverse is true in $F'$. Hence, LICC rules out nonlocal causal influences from the settings to the outcomes, while allowing for frame-dependent, nonlocal causal connections between the measurement outcomes $A$ and $B$. This leaves us with the following pair of causal structures:
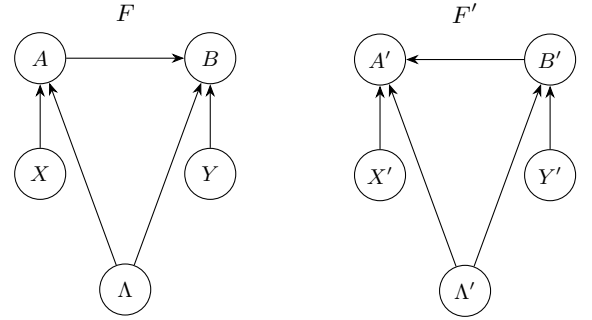
\begin{figure}[H]
\centering
\begin{minipage}{0.48\columnwidth}
\centering
$F$\\[0.8ex]
\begin{tikzpicture}[causal graph]
 
  \node[draw, circle, minimum size=0.75cm, inner sep=0pt] (a) at (0,0) {$A$};
  \node[draw, circle, minimum size=0.75cm, inner sep=0pt] (b) at (2.5,0) {$B$};
  \node[draw, circle, minimum size=0.75cm, inner sep=0pt] (x) at (0,-1.7) {$X$};
  \node[draw, circle, minimum size=0.75cm, inner sep=0pt] (y) at (2.5,-1.7) {$Y$};
  \node[draw, circle, minimum size=0.75cm, inner sep=0pt] (lambda) at (1.25,-3.4) {$\Lambda$};

  \draw[->] (a) -- (b);
  \draw[->] (x) -- (a);
  \draw[->] (y) -- (b);
  \draw[->] (lambda) -- (a);
  \draw[->] (lambda) -- (b);
  
\end{tikzpicture}
\end{minipage}
\hfill
\begin{minipage}{0.48\columnwidth}
\centering
$F'$\\[0.8ex]
\begin{tikzpicture}[causal graph]

  \node[draw, circle, minimum size=0.75cm, inner sep=0pt] (a) at (0,0) {$A'$};
  \node[draw, circle, minimum size=0.75cm, inner sep=0pt] (b) at (2.5,0) {$B'$};
  \node[draw, circle, minimum size=0.75cm, inner sep=0pt] (x) at (0,-1.7) {$X'$};
  \node[draw, circle, minimum size=0.75cm, inner sep=0pt] (y) at (2.5,-1.7) {$Y'$};
  \node[draw, circle, minimum size=0.75cm, inner sep=0pt] (lambda) at (1.25,-3.4) {$\Lambda'$};

  \draw[->] (b) -- (a);
  \draw[->] (x) -- (a);
  \draw[->] (y) -- (b);
  \draw[->] (lambda) -- (a);
  \draw[->] (lambda) -- (b);

\end{tikzpicture}
\end{minipage}
\caption{The maximally allowed causal structure underlying the Bell experiment compatible with IS, NRC, and LICC in $F$ and $F'$.}
\end{figure}
Applying the Markov condition, we find that the operational statistics must take the following form:
\begin{align}
    p(a, b \mid x, y) &=  \int \mathrm{d}\lambda \,p(a \mid x, \lambda)\,p(b \mid a, y, \lambda)\,p(\lambda)\\
    &= \int \mathrm{d}\lambda' \,p'(a' \mid b', x', \lambda')\,p'(b' \mid  y', \lambda')\,p'(\lambda')\,,
\end{align}
That the operational probabilities are identical in both frames is required by Runwise Agreement. For $|A| = |B| = 2$ and $|X| = |Y| = 3$, it has been shown \cite{Chaves_2015, Ringbauer_2016} that both factorizations satisfy the same inequality:
\begin{equation}
   E_{00} - E_{02} - E_{11} + E_{12} - E_{20} + E_{21} \leq 4\,.
\end{equation}
This inequality was experimentally violated up to a value of $5.16 \pm 0.02$, which corresponds to a violation by more than $170$ standard deviations \cite{Ringbauer_2016}. Hence, any causal model satisfying the conjunction of IS, NRC, and LICC is unable to reproduce the operational statistics and is therefore empirically inadequate.
\end{proof}

\section{Conclusion}\label{sec:V}
We have shown that the idea underlying Gisin's argument can be generalized to the genuinely probabilistic case using frame-indexed causal models. Within this framework, we established two no-go theorems according to which no empirically adequate frame-indexed causal model of a Bell experiment can jointly satisfy:
\begin{itemize}
    \item[\ref{theorem:1}] Independent Settings, No-Retrocausality, and Causal Lorentz Invariance;
    \item[\ref{theorem:2}] Independent Settings, No-Retrocausality, and Lorentz Invariance of Causal Connections.
\end{itemize}
Theorem~\ref{theorem:2} strictly strengthens Theorem~\ref{theorem:1}, since LICC is strictly weaker than CLI. Nevertheless, Theorem~\ref{theorem:1} should not be regarded as redundant. It provides the most straightforward extension of Gisin's reasoning to the probabilistic case and offers a natural route to Theorem~\ref{theorem:2}. The reasoning underlying the proof of Theorem~\ref{theorem:1} is also used to help establish an additional no-go result found in the Appendix.

The significance of Theorem~\ref{theorem:2} is that LICC allows the direction of causal influence to depend on the choice of reference frame. LICC requires only the skeleton of the causal structure to be frame-independent. The no-go result therefore excludes a class of causal explanations of the Bell scenario in which nonlocal causal influence is retained but its direction is consistent with the frame-dependent temporal ordering, as enforced by NRC. This is particularly relevant to proposals motivated by the notion of effective causality in stochastic relativistic theories such as rGRWf.

One possible response to Theorem~\ref{theorem:2}, while retaining IS and NRC, is to abandon invariance requirements on causal structure altogether. This raises further questions regarding the physical status of causal structure. If substantially different causal structures are associated with the same experimental situation in different inertial frames, it becomes less clear in what sense the causal structure can be seen as ontic, as for example advocated by Spekkens \cite{Spekkens_2015}.

The additional result presented in the Appendix explores this possibility. There we abandon all invariance assumptions on causal structure and instead impose conditions on the probabilistic relation between distinct frame-indexed causal explanations. If differing causal structures are interpreted as representing distinct mechanisms, then Independent Explanation requires their outcomes to be statistically independent once the respective settings and latent variables have been specified. Together with Runwise Agreement, Independent Explanation forces determinism in order to explain the perfect agreement between frames. This then allows for Gisin's original argument to go through.

A different response is to deny that the frame-indexed causal structures represent distinct physical mechanisms. If causal structure is instead regarded as a useful representation of a more fundamental stochastic process, then there is considerably less motivation for requiring Independent Explanation to hold between different frame-indexed descriptions. Tumulka's position appears to be close to this view. Although his discussion of effective causality can be read as permitting frame-dependent directions of influence, he also emphasizes that a fundamental relativistic theory does not need to specify any direction of influence at all:
\begin{quote}
    Now the claim is that the theory [rGRWf] does not have to specify which of the two ways nature uses and that nature does not have to use either of the two ways. [\dots] Thus, there is no need for a direction of influence; it is enough if the fundamental physical theory prescribes the joint distribution. \cite{Tumulka_2022}
\end{quote}
This way of understanding the role of causality in rGRWf and related theories evades the no-go results presented here by refusing the applicability of causal notions to fundamental physics.

In the case of relativistic extensions of Bohmian mechanics, one does not need to go so far. Relativistic versions of Bohmian mechanics introduce a preferred foliation of spacetime relative to which the dynamics are specified \cite{Duerr_2009, Duerr_2014}. Causal influences between spacelike separated variables that are future-directed relative to the preferred foliation can be retrocausal in other frames. They therefore evade the present no-go results by rejecting NRC. The introduction of such additional spacetime structure is, however, often regarded as being in tension with the spirit of special relativity \cite{Maudlin_2008}.

Finally, the present discussion may also be relevant beyond classical causal models. One powerful motivation for nonclassical causal inference is the possibility of providing causal explanations for correlations that violate Bell locality without violating faithfulness, for example by allowing nonlocal causal influences \cite{Wood_2015, Costa_2016, Allen_2017}. However, recent work by Y\={\i}ng et al. on extended Wigner's friend scenarios shows that under an assumption of Absoluteness of Observed Events, i.e., that there exists a joint probability distribution over all random variables operationally accessible to some observer, even very general classes of nonclassical causal models, such as $d$-separated or compositional causal models, must violate faithfulness \cite{Ying_2024}. This suggests that Gisin-style arguments can be extended beyond the classical causal models framework. 

Taken together, the presented results clarify an aspect of the tension between nonlocality and relativity. They do so by showing that explanations of Bell correlations cannot respect natural invariance assumptions on causal structure while maintaining other reasonable physical principles, such as No-Retrocausality and the freedom of the experimenters to independently choose their measurement settings.

\section*{Acknowledgments}\label{acknowledgments}
I am grateful to Beatrix C. Hiesmayr for supervising my master's thesis, on which this paper is in part based. I also thank Flavio Del Santo and Benjamin Kesetovic for helpful discussions.

\appendix\label{appendix}
\section{Gisin's Argument in the Ontological Models Framework}\label{appendix:A}
In the following, the probabilistic extension of Gisin's argument using ontological models, briefly sketched in Subsection~\ref{subsec:IIIB}, is formulated as a no-go theorem. As mentioned above, this result does not provide a satisfactory generalization of Gisin's argument, since the assumption of PNRC is highly restrictive.

In order to infer Bell locality in this manner, one has to assume that the same statistical independence relations hold in all reference frames:
\begin{definition}[Probabilistic Lorentz Invariance]
    The set of statistical independence relations is identical in all reference frames.
\end{definition}
Using this assumption of \emph{Probabilistic Lorentz Invariance} (PLI), the result can be formulated in the following way:
\begin{theorem}\label{theorem:3}
 There is no empirically adequate, ontological model of the Bell scenario that satisfies:
    \begin{itemize}
        \item[(i)] Measurement Independence,
        \item[(ii)] Probabilistic No-Retrocausality,
        \item[(iii)] Probabilistic Lorentz Invariance.
    \end{itemize}
\end{theorem}
\begin{proof}
Assume an ontological model that accounts for the observed statistics of the Bell experiment. By the existence of an ontological model, MI and the definition of conditional probability we obtain
\begin{align}
    p(a, b \mid x, y) 
    &= \int \mathrm{d}\lambda\,p(a, b \mid x, y, \lambda)\,p(\lambda \mid x, y)\\
    &= \int \mathrm{d}\lambda\,p(a \mid b, x, y, \lambda)\, p(b \mid x, y, \lambda)\,p(\lambda)\,.
\end{align}
We consider a reference frame $F$ in which Alice performs her experiment before Bob. In this frame, PNRC implies 
\begin{equation}
    p(a \mid b, x, y, \lambda) = p(a \mid x, \lambda)\,,
\end{equation}
because in $F$, $B$ and $Y$ are in the future relative to $A$, while $X$ and $\Lambda$ specify the past of $A$. In $F'$, a frame in which Bob's experiment takes place before Alice's, PNRC implies for analogous reasons that
\begin{equation}
    p'(b' \mid x', y', \lambda') = p'(b' \mid y', \lambda')\,.
\end{equation}
By PLI, the independence relation obtained in $F'$ must also hold in $F$. Hence, we conclude that $p(a, b \mid x, y)$ is Bell-local:
\begin{align}
    p(a, b \mid x, y) 
    &= \int \mathrm{d}\lambda\,p(a \mid b, x, y, \lambda)\, p(b \mid x, y, \lambda)\,p(\lambda)\\
    &= \int \mathrm{d}\lambda\,p(a \mid x, \lambda)\, p(b \mid y, \lambda)\,p(\lambda)\,.
\end{align}
Thus, any ontological model satisfying MI, PNRC, and PLI cannot violate the CHSH inequality and is consequently empirically inadequate. 
\end{proof}

\section{Frame-Indexed Causal Models}\label{appendix:B}
As mentioned in Section~\ref{sec:IV}, the motivation for the framework of frame-indexed causal models comes from studying possibly differing causal explanations of the same physical experiment from different reference frames. The notion of frame-indexed causal models consists of the following elements:
\begin{itemize}
    \item[(i)] A set of reference frames $\mathcal{F} = \{F_1, \dots, F_m\}$.
    \item[(ii)] A set of causal models $\mathcal{C}$, where exactly one causal model $C_j \in \mathcal{C}$ is assigned to each reference frame $F_j \in \mathcal{F}$. Furthermore, all corresponding variables have the same latency status.
    \item[(iii)] A set of spacetime regions $\mathcal{R}$, where for every frame $F_j$ and corresponding causal model $C_j$, the variable $X_i^{F_j}$ is uniquely associated with the region $R_i$, for all $i$.
    \item[(iv)] A joint probability distribution $P$ over all random variables of all causal models, meaning $P = P\left(\mathbf{V^{F_1}}, \dots,\mathbf{V^{F_m}}\right)$. It is assumed that for nonlatent variables, the random variable $X_i^{F_j}$ takes the same value for every run as $X_i^{F_k}$ for all $(i, j, k)$, meaning that $P(X_i^{F_j} = X_i^{F_k}) = 1$.
\end{itemize}
In order to allow for possibly differing causal explanations of the same experimental situation we need to have a causal model associated with each frame under consideration as stipulated in (i) and (ii). But in order for these possibly differing causal explanations to describe the same experimental situation, the corresponding variables of the different causal models must be located in spacetime in accordance with the experiment as ensured by (iii). Furthermore, they must reproduce the same operational values for every run of the experiment, and thereby the same operational statistics, as assumed by (iv). 

One could simply identify all corresponding random variables across different reference frames and thereby work with a single set of random variables rather than separate frame-indexed copies. There are two reasons why we refrain from doing so: Firstly, we only want to impose invariance assumptions which follow from the requirement that the frame-indexed models describe the same physical experiment, such as Runwise Agreement. Secondly, a joint distribution over the frame-indexed variables is required in order to formulate probabilistic relations between variables belonging to different frame-indexed causal models, including Runwise Agreement and Independent Explanation. These assumptions will be vital for the no-go argument given in Appendix~\ref{appendix:C}. 

Now that the elements which comprise the notion of frame-indexed causal models have been introduced and discussed, we are able to give the full definition. We begin by defining \emph{Runwise Agreement} (RA):
\begin{definition}[Runwise Agreement]
Let $C_1$ and $C_2$ be causal models with sets of nonlatent variables $\mathbf{V}^1 = \{X_1^1, \dots, X_n^1\}$ and $\mathbf{V}^2 = \{X_1^2, \dots, X_n^2\}$, respectively, which are jointly distributed according to a probability distribution $P$ over $\mathbf{V}^1 \cup \mathbf{V}^2$. The random variables satisfy
\begin{equation}
P\left(X_i^1 = X_i^2\right) = 1\,,
\end{equation}
for every $i \in \{1,\dots,n\}$.
\end{definition}
RA is substantially stronger than observational equivalence. Observational equivalence requires only that two models reproduce the same probability distribution over the operational variables. RA requires the corresponding variables to take the same values in every individual run. It therefore implies observational equivalence, whereas the converse does not generally hold. Using this notion of RA, frame-indexed causal models can be defined as follows:
\begin{definition}[Frame-Indexed Causal Model]
A \emph{frame-indexed causal model} is a quadruple
\begin{equation}
    \mathfrak{C} = \bigl(\mathcal{F},\mathcal{C},\mathcal{R},P\bigr),
\end{equation}
where $\mathcal{F} = \{F_1,\dots,F_m\}$ is a set of inertial reference frames and $\mathcal{C} = \{C^{F_1}, \dots, C^{F_m}\}$ is a corresponding set of causal models. For each frame $F_j\in\mathcal{F}$, the corresponding causal model
\begin{equation}
C^{F_j} = \bigl(G^{F_j},P^{F_j}\bigr)\,,
\end{equation}
is defined over the frame-indexed set of random variables $\mathbf{V}^{F_j} = \{X_1^{F_j},\ldots,X_n^{F_j}\}$, with $\left|\mathbf{V}^{F_j}\right| = n$ for all frames. Corresponding variables are assumed to have the same latency status in every frame; that is, $X^{F_j}_i$ is latent if and only if $X^{F_k}_i$ is latent for all $(j, k)$. The set $\mathcal{R} = \{R_1,\ldots,R_n\}$ is a set of spacetime regions such that, for every frame $F_j$, the variable $X_i^{F_j}$ is associated with the region $R_i$ for all $i \in \{1, \dots, n\}$. Furthermore, $P = P\bigl(\mathbf{V}^{F_1},\dots,\mathbf{V}^{F_m}\bigr)$ is a joint probability distribution over all frame-indexed random variables. For each frame $F_j$, the distribution $P^{F_j}$ is the marginal of $P$ over the variables in $\mathbf{V}^{F_j}$. Each marginal $P^{F_j}$ is assumed to be Markov relative to $G^{F_j}$. Finally, RA is assumed to hold between all pairs of causal models in $\mathcal{C}$.
\end{definition}
This notion of frame-indexed causal models imposes no constraint on how the causal structure $G^{F_j}$ depends on the choice of reference frame.

\section{Determinism Inferred}\label{appendix:C}
The main part of this paper leaves open the possibility that empirically adequate frame-indexed causal models of the Bell scenario respect NRC and IS while allowing for arbitrary frame-dependence with regard to causal structure. In the following we investigate the possibility of frame-dependent causal explanations more broadly. We will see that taking the idea of frame-dependent explanation seriously imposes a restriction to deterministic causal models of the Bell scenario. A causal model of the Bell scenario is said to be deterministic if the outcomes $A$, $B$ are fixed by the settings $X$, $Y$ and the hidden variable $\Lambda$:
\begin{equation}
    p(a, b \mid x, y, \lambda) = \delta_{(a, b), h(x, y, \lambda)} \in \{0, 1\}\,.
\end{equation}
Recovering determinism enables a direct application of Gisin's original reasoning.

In order to make this argument work, we will assume a principle which states that if in a frame-indexed causal model of the Bell scenario, two causal models do not have the same causal structure, then their respective outcomes must be statistically independent given their respective settings and the latent variables. The motivation behind this assumption, referred to as \emph{Independent Explanation} (IE), becomes most evident when considering the converse situation, namely when two causal models of the same Bell experiment have the same causal structure. In this case, one can argue that the causal models of the experiment describe one and the same mechanism by which the outcomes of the experiments arise from different frames. This mechanism can therefore be taken to be independent of the choice of reference frame. Hence, it would be unreasonable to demand an explanation for the agreement between the outcome of the experiments in different reference frames. 

Now consider the situation in the case of differing causal structures. Because of the difference between causal structures it seems reasonable to assume that the experimental outcomes arise by way of a different mechanism. But if different and hence distinct mechanisms each produce the experimental outcomes, the outcomes should be statistically independent from each other, given their respective inputs. This intuition can be more formally expressed as follows:
\begin{definition}[Independent Explanation]
    If a frame-indexed causal model of a Bell experiment exhibits different causal structures, $G^F \neq G^{F'}$, in different frames $F$ and $F'$, then, conditional on their respective measurement settings and latent variables, their measurement outcomes are statistically independent:
      \begin{align*}
        &P(A, A', B, B' \mid X, X', Y, Y', \Lambda, \Lambda') =\\ &P(A, B \mid X, Y, \Lambda)\,P(A', B' \mid X', Y', \Lambda')\,.
    \end{align*}
\end{definition}
The idea behind the argument presented below is the following. Either the causal structure underlying the Bell experiment is the same in all frames or not. If it is the same, then it is, by the reasoning of Theorem~\ref{theorem:1}, unable to explain violations of the CHSH inequality. If it is different in at least a pair of frames, then we need some explanation why two distinct mechanisms, as demanded by IE, are able to produce the same results in both frames. It will turn out that this is only possible if both causal models are deterministic. But, as mentioned before, for the case of deterministic models, Gisin's original argument applies. As discussed in Section~\ref{sec:II}, Gisin's argument requires \emph{Cross-Frame Measurement Independence}:
\begin{definition}[Cross-Frame Measurement Independence]
 For a frame-indexed causal model of the Bell scenario, the hidden variables $\Lambda$ and $\Lambda'$, associated with frames $F$ and $F'$, are jointly statistically independent of the settings:
    \begin{equation}
        P(\Lambda, \Lambda' \mid X, Y) = P(\Lambda, \Lambda')\,.
    \end{equation}
\end{definition}
This enables us to establish the following no-go result: 
\begin{theorem}\label{theorem:4}
     There is no empirically adequate, frame-indexed causal model of a Bell experiment that satisfies:
    \begin{itemize}
        \item[(i)] Independent Settings,
        \item[(ii)] No-Retrocausality,
        \item[(iii)] Independent Explanation,
        \item[(iv)] Cross-Frame Measurement Independence. 
    \end{itemize}
\end{theorem}
\begin{proof}
As before, by assuming IS and NRC and considering the causal structure in $F$ and $F'$ we obtain the following, frame-dependent pair of causal structures:
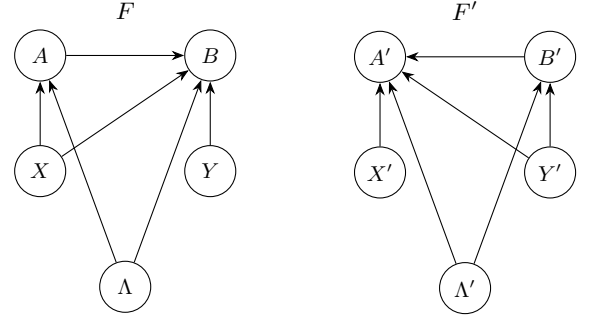
\begin{figure}[H]
\centering
\begin{minipage}{0.48\columnwidth}
\centering
$F$\\[0.8ex]
\begin{tikzpicture}[causal graph]
 
  \node[draw, circle, minimum size=0.75cm, inner sep=0pt] (a) at (0,0) {$A$};
  \node[draw, circle, minimum size=0.75cm, inner sep=0pt] (b) at (2.5,0) {$B$};
  \node[draw, circle, minimum size=0.75cm, inner sep=0pt] (x) at (0,-1.7) {$X$};
  \node[draw, circle, minimum size=0.75cm, inner sep=0pt] (y) at (2.5,-1.7) {$Y$};
  \node[draw, circle, minimum size=0.75cm, inner sep=0pt] (lambda) at (1.25,-3.4) {$\Lambda$};

  \draw[->] (a) -- (b);
  \draw[->] (x) -- (a);
  \draw[->] (x) -- (b);
  \draw[->] (y) -- (b);
  \draw[->] (lambda) -- (a);
  \draw[->] (lambda) -- (b);
  
\end{tikzpicture}
\end{minipage}
\hfill
\begin{minipage}{0.48\columnwidth}
\centering
$F'$\\[0.8ex]
\begin{tikzpicture}[causal graph]

  \node[draw, circle, minimum size=0.75cm, inner sep=0pt] (a) at (0,0) {$A'$};
  \node[draw, circle, minimum size=0.75cm, inner sep=0pt] (b) at (2.5,0) {$B'$};
  \node[draw, circle, minimum size=0.75cm, inner sep=0pt] (x) at (0,-1.7) {$X'$};
  \node[draw, circle, minimum size=0.75cm, inner sep=0pt] (y) at (2.5,-1.7) {$Y'$};
  \node[draw, circle, minimum size=0.75cm, inner sep=0pt] (lambda) at (1.25,-3.4) {$\Lambda'$};

  \draw[->] (b) -- (a);
  \draw[->] (x) -- (a);
  \draw[->] (y) -- (b);
  \draw[->] (y) -- (a);
  \draw[->] (lambda) -- (a);
  \draw[->] (lambda) -- (b);

\end{tikzpicture}
\end{minipage}
\caption{The maximally allowed causal structure underlying the Bell experiment compatible with IS and NRC in $F$ and $F'$.}
\end{figure}
The arrows drawn in the figure above are those which are allowed by IS and NRC, but none of the arrows are enforced by the current assumptions. Hence, the actual causal structure underlying the Bell experiment may have fewer arrows than depicted here. In any case, either the causal structure is the same in both frames or it differs between frames. If it is the same in both frames, then the reasoning used in the proof of Theorem~\ref{theorem:1} applies, and the resulting causal model is empirically inadequate. We therefore only have to consider the case $G^F \neq G^{F'}$. By IE we have that
\begin{equation}
    P(U, U' \mid V, V') = P(U \mid V)\,P(U' \mid V')\,.
\end{equation}
where $U = (A, B)$, $V = (X, Y, \Lambda)$ and analogously for $U'$ and $V'$.
By RA: 
\begin{equation}
    P(U = U') = 1\,,
\end{equation}
as $U$ and $U'$ are nonlatent. Hence, for almost all $(v, v')$
\begin{equation}
    P(U = U' \mid v, v') = 1\,.
\end{equation}
Therefore,
\begin{align}
    1 &= P(U = U' \mid v, v')\\
    &= \sum_u  P(U = u, U'= u \mid v, v')\\
    &= \sum_u P(U = u\mid v)\,P(U' = u\mid v')\,.
\end{align}
Since both factors are normalized conditional probability distributions,
\begin{equation}
    \sum_u P(U=u\mid v)P(U'=u\mid v')
    \leq
    \max_u P(U'=u\mid v')
    \leq 1\,,
\end{equation}
with equality only if there exists some $u_0$ such that
\begin{equation}
    P(U'=u_0\mid v') = 1\,.
\end{equation}
This further implies that
\begin{equation}
    P(U = u_0 \mid v) = 1\,.
\end{equation}
Thus both causal models are deterministic, meaning that
\begin{equation}
    p(a, b \mid x, y, \lambda) = \delta_{a, f(x, y, \lambda)}\,\delta_{b, g(x, y, \lambda)} \in \{0, 1\}\,,
\end{equation}
almost surely and analogously for $p'(a', b' \mid x', y', \lambda')$. This allows us to run Gisin's original argument as outlined in Section~\ref{sec:II}. In $F$ we have 
\begin{equation}
     a = f(x, y, \lambda) = f(x, \lambda)\,,
\end{equation}
while in $F'$ we find that
\begin{equation}
     b' = g'(x', y', \lambda') = g'(y', \lambda')\,,
\end{equation}
Hence, by RA, there exists the following local model of the Bell experiment
\begin{align}
    a &= f(x, \lambda)\,,\\
    b &= g'(y, \lambda')\,.
\end{align}
The operational statistics $p(a, b \mid x, y)$ can then, using Cross-Frame Measurement Independence, be expressed in the following way:
\begin{align}
    p(a, b \mid x, y) &= \int \mathrm{d}\lambda\,\mathrm{d}\lambda' \,\delta_{a, f(x, \lambda)}\,\delta_{b, g'(y, \lambda')}\,p(\lambda, \lambda' \mid x, y)\\
    &= \int \mathrm{d}\lambda\,\mathrm{d}\lambda' \,\delta_{a, f(x, \lambda)}\,\delta_{b, g'(y, \lambda')}\,p(\lambda, \lambda')\,.
\end{align}
This behavior is Bell-local and therefore unable to violate the CHSH inequality as observed in experiments, making the causal model empirically inadequate. 
\end{proof}
Unlike Theorems~\ref{theorem:1} and~\ref{theorem:2}, Theorem~\ref{theorem:4} does not impose any invariance assumption on causal structure. It instead constrains the probabilistic coordination between frames with differing causal structures.

The plausibility of IE seems to hinge on the physical status of causal structure. If causal structure is taken to be ontic, then differing causal structures naturally suggest the existence of distinct mechanisms by which the operational outcomes are generated. However, if different frame-indexed causal structures are understood merely as alternative representations of a single underlying stochastic process, then Independent Explanation seems drastically less plausible. 

\bibliography{references}

@article{Gisin_2011,
  author  = {Gisin, Nicolas},
  title   = {Impossibility of Covariant Deterministic Nonlocal Hidden-Variable Extensions of Quantum Theory},
  journal = {Physical Review A},
  volume  = {83},
  number  = {2},
  pages   = {020102},
  year    = {2011},
  doi     = {10.1103/PhysRevA.83.020102}
}

@article{Esfeld_2013,
  author  = {Esfeld, Michael and Gisin, Nicolas},
  title   = {The {GRW} Flash Theory: A Relativistic Quantum Ontology of Matter in Space-Time?},
  journal = {Philosophy of Science},
  volume  = {81},
  number  = {2},
  pages   = {248--264},
  year    = {2014},
  doi     = {10.1086/675730},
  eprint  = {1310.5308},
  archivePrefix = {arXiv},
  primaryClass  = {quant-ph}
}

@incollection{Levy_2020,
  author    = {Levy, Avi and Hemmo, Meir},
  title     = {Why a Relativistic Quantum Mechanical World Must Be Indeterministic},
  booktitle = {Quantum, Probability, Logic: The Work and Influence of Itamar Pitowsky},
  editor    = {Hemmo, Meir and Shenker, Orly},
  series    = {Jerusalem Studies in Philosophy and History of Science},
  pages     = {423--447},
  publisher = {Springer},
  address   = {Cham},
  year      = {2020},
  doi       = {10.1007/978-3-030-34316-3_19}
}

@article{Bell_1988,
  author  = {Mann, Charles and Crease, Robert P.},
  title   = {Interview: John Bell},
  journal = {Omni},
  volume  = {10},
  number  = {8},
  pages   = {84--86, 88, 90, 92, 121},
  month   = may,
  year    = {1988}
}

@book{Maudlin_2011,
  author    = {Maudlin, Tim},
  title     = {Quantum Non-Locality and Relativity: Metaphysical Intimations of Modern Physics},
  edition   = {3},
  publisher = {Wiley-Blackwell},
  year      = {2011},
  doi       = {10.1002/9781444396973},
  isbn      = {9781444331264}
}

@incollection{Maudlin_2008,
  author    = {Maudlin, Tim},
  title     = {Non-Local Correlations in Quantum Theory: How the Trick Might Be Done},
  booktitle = {Einstein, Relativity and Absolute Simultaneity},
  editor    = {Craig, William Lane and Smith, Quentin},
  pages     = {156--179},
  publisher = {Routledge},
  year      = {2008}
}

@article{Ghirardi_2010,
  author  = {Ghirardi, GianCarlo},
  title   = {Does Quantum Nonlocality Irremediably Conflict with Special Relativity?},
  journal = {Foundations of Physics},
  volume  = {40},
  number  = {9--10},
  pages   = {1379--1395},
  year    = {2010},
  doi     = {10.1007/s10701-009-9391-9}
}

@incollection{Myrvold_2021,
  author    = {Myrvold, Wayne C.},
  title     = {Relativistic Constraints on Interpretations of Quantum Mechanics},
  booktitle = {The Routledge Companion to Philosophy of Physics},
  editor    = {Knox, Eleanor and Wilson, Alastair},
  pages     = {99--121},
  publisher = {Routledge},
  year      = {2021},
  eprint    = {2107.02089},
  archivePrefix = {arXiv},
  primaryClass  = {quant-ph}
}

@article{Chaves_2015,
  author  = {Chaves, Rafael and Kueng, Richard and Brask, Jonatan Bohr and Gross, David},
  title   = {Unifying Framework for Relaxations of the Causal Assumptions in Bell's Theorem},
  journal = {Physical Review Letters},
  volume  = {114},
  number  = {14},
  pages   = {140403},
  year    = {2015},
  month   = apr,
  doi     = {10.1103/PhysRevLett.114.140403}
}

@article{Ringbauer_2016,
  author  = {Ringbauer, Martin and Giarmatzi, Christina and Chaves, Rafael and Costa, Fabio and White, Andrew G. and Fedrizzi, Alessandro},
  title   = {Experimental Test of Nonlocal Causality},
  journal = {Science Advances},
  volume  = {2},
  number  = {8},
  pages   = {e1600162},
  year    = {2016},
  month   = aug,
  doi     = {10.1126/sciadv.1600162}
}

@article{Clauser_1969,
  author  = {Clauser, John F. and Horne, Michael A. and Shimony, Abner and Holt, Richard A.},
  title   = {Proposed Experiment to Test Local Hidden-Variable Theories},
  journal = {Physical Review Letters},
  volume  = {23},
  number  = {15},
  pages   = {880--884},
  year    = {1969},
  month   = oct,
  doi     = {10.1103/PhysRevLett.23.880}
}

@article{Tsirelson_1980,
  author  = {Cirel'son, Boris S.},
  title   = {Quantum Generalizations of {Bell}'s Inequality},
  journal = {Letters in Mathematical Physics},
  volume  = {4},
  number  = {2},
  pages   = {93--100},
  year    = {1980},
  month   = mar,
  doi     = {10.1007/BF00417500}
}

@article{Levy_2025,
  author  = {Levy, Avi and Hemmo, Meir},
  title   = {Locality and Probability in Relativistic Quantum Theories and Hidden Variables Quantum Theories},
  journal = {Foundations of Physics},
  volume  = {55},
  number  = {5},
  pages   = {73},
  year    = {2025},
  doi     = {10.1007/s10701-025-00885-8},
  eprint  = {2503.12967},
  archivePrefix = {arXiv},
  primaryClass  = {quant-ph}
}

@article{Giustina_2015,
  author = {Giustina, Marissa and Versteegh, Marijn A. M. and Wengerowsky, S{\"o}ren
            and Handsteiner, Johannes and Hochrainer, Armin and Phelan, Kevin
            and Steinlechner, Fabian and Kofler, Johannes and Larsson, Jan-{\AA}ke
            and Abell{\'a}n, Carlos and Amaya, Waldimar and Pruneri, Valerio
            and Mitchell, Morgan W. and Beyer, J{\"o}rn and Gerrits, Thomas
            and Lita, Adriana E. and Shalm, Lynden K. and Nam, Sae Woo
            and Scheidl, Thomas and Ursin, Rupert and Wittmann, Bernhard
            and Zeilinger, Anton},
  title   = {Significant-Loophole-Free Test of {Bell}'s Theorem with Entangled Photons},
  journal = {Physical Review Letters},
  volume  = {115},
  number  = {25},
  pages   = {250401},
  year    = {2015},
  month   = dec,
  doi     = {10.1103/PhysRevLett.115.250401}
}

@article{Jarrett_1984,
  author  = {Jarrett, Jon P.},
  title   = {On the Physical Significance of the Locality Conditions in the {Bell} Arguments},
  journal = {No{\^u}s},
  volume  = {18},
  number  = {4},
  pages   = {569--589},
  year    = {1984},
  doi     = {10.2307/2214878}
}

@book{Scarani_2019,
  author    = {Scarani, Valerio},
  title     = {Bell Nonlocality},
  publisher = {Oxford University Press},
  year      = {2019},
  doi       = {10.1093/oso/9780198788416.001.0001},
  isbn      = {9780198788416}
}

@article{Tumulka_2006,
  author  = {Tumulka, Roderich},
  title   = {A Relativistic Version of the {Ghirardi--Rimini--Weber} Model},
  journal = {Journal of Statistical Physics},
  volume  = {125},
  number  = {4},
  pages   = {821--840},
  year    = {2006},
  doi     = {10.1007/s10955-006-9227-3}
}

@article{Tumulka_2007,
  author  = {Tumulka, Roderich},
  title   = {Comment on ``The Free Will Theorem''},
  journal = {Foundations of Physics},
  volume  = {37},
  number  = {2},
  pages   = {186--197},
  year    = {2007},
  doi     = {10.1007/s10701-006-9098-0}
}

@article{Allori_2022,
  author  = {Allori, Valia},
  title   = {What Is It Like to Be a Relativistic {GRW} Theory? Or: Quantum Mechanics and Relativity, Still in Conflict After All These Years},
  journal = {Foundations of Physics},
  volume  = {52},
  number  = {4},
  pages   = {79},
  year    = {2022},
  doi     = {10.1007/s10701-022-00595-5}
}

@article{Duerr_2014,
  author  = {D{\"u}rr, Detlef and Goldstein, Sheldon and Norsen, Travis and Struyve, Ward and Zangh{\`i}, Nino},
  title   = {Can {Bohmian} Mechanics Be Made Relativistic?},
  journal = {Proceedings of the Royal Society A: Mathematical, Physical and Engineering Sciences},
  volume  = {470},
  number  = {2162},
  pages   = {20130699},
  year    = {2014},
  doi     = {10.1098/rspa.2013.0699},
  eprint  = {1307.1714},
  archivePrefix = {arXiv},
  primaryClass  = {quant-ph}
}

@book{Duerr_2009,
  author    = {D{\"u}rr, Detlef and Teufel, Stefan},
  title     = {Bohmian Mechanics: The Physics and Mathematics of Quantum Theory},
  publisher = {Springer},
  address   = {Berlin, Heidelberg},
  year      = {2009},
  doi       = {10.1007/b99978},
  isbn      = {978-3-540-89343-1}
}

@misc{Gisin_2010,
  author        = {Gisin, Nicolas},
  title         = {The Free Will Theorem, Stochastic Quantum Dynamics and
                   True Becoming in Relativistic Quantum Physics},
  year          = {2010},
  eprint        = {1002.1392},
  archiveprefix = {arXiv},
  primaryclass  = {quant-ph},
  doi           = {10.48550/arXiv.1002.1392},
  url           = {https://arxiv.org/abs/1002.1392}
}

@incollection{Spekkens_2015,
  author    = {Spekkens, Robert W.},
  title     = {The Paradigm of Kinematics and Dynamics Must Yield to Causal Structure},
  booktitle = {Questioning the Foundations of Physics: Which of Our Fundamental Assumptions Are Wrong?},
  editor    = {Aguirre, Anthony and Foster, Brendan and Merali, Zeeya},
  series    = {The Frontiers Collection},
  pages     = {5--16},
  publisher = {Springer},
  address   = {Cham},
  year      = {2015},
  doi       = {10.1007/978-3-319-13045-3_2},
  eprint    = {1209.0023},
  archivePrefix = {arXiv},
  primaryClass  = {quant-ph}
}

@article{Vilasini_2022,
  author  = {Vilasini, V. and Colbeck, Roger},
  title   = {General framework for cyclic and fine-tuned causal models and their compatibility with space-time},
  journal = {Physical Review A},
  volume  = {106},
  pages   = {032204},
  year    = {2022},
  doi     = {10.1103/PhysRevA.106.032204},
  eprint  = {2109.12128},
  archivePrefix = {arXiv},
  primaryClass  = {quant-ph}
}

@article{Wood_2015,
  author  = {Wood, Christopher J. and Spekkens, Robert W.},
  title   = {The Lesson of Causal Discovery Algorithms for Quantum Correlations: Causal Explanations of {Bell}-Inequality Violations Require Fine-Tuning},
  journal = {New Journal of Physics},
  volume  = {17},
  number  = {3},
  pages   = {033002},
  year    = {2015},
  doi     = {10.1088/1367-2630/17/3/033002},
  eprint  = {1208.4119},
  archivePrefix = {arXiv},
  primaryClass  = {quant-ph}
}

@article{Ying_2024,
  author  = {Yīng, Y{\`i}l{\`e} and Ansanelli, Marina Maciel and Di Biagio, Andrea and Wolfe, Elie and Schmid, David and Cavalcanti, Eric Gama},
  title   = {Relating {Wigner}'s Friend Scenarios to Nonclassical Causal Compatibility, Monogamy Relations, and Fine Tuning},
  journal = {Quantum},
  volume  = {8},
  pages   = {1485},
  year    = {2024},
  month   = sep,
  doi     = {10.22331/q-2024-09-26-1485},
  eprint  = {2309.12987},
  archivePrefix = {arXiv},
  primaryClass  = {quant-ph}
}

@article{Harrigan_2010,
  author  = {Harrigan, Nicholas and Spekkens, Robert W.},
  title   = {Einstein, Incompleteness, and the Epistemic View of Quantum States},
  journal = {Foundations of Physics},
  volume  = {40},
  number  = {2},
  pages   = {125--157},
  year    = {2010},
  doi     = {10.1007/s10701-009-9347-0},
  eprint  = {0706.2661},
  archivePrefix = {arXiv},
  primaryClass  = {quant-ph}
}

@article{Allen_2017,
  author  = {Allen, John-Mark A. and Barrett, Jonathan and Horsman, Dominic C. and Lee, Ciar{\'a}n M. and Spekkens, Robert W.},
  title   = {Quantum Common Causes and Quantum Causal Models},
  journal = {Physical Review X},
  volume  = {7},
  number  = {3},
  pages   = {031021},
  year    = {2017},
  doi     = {10.1103/PhysRevX.7.031021}
}

@article{Costa_2016,
  author  = {Costa, Fabio and Shrapnel, Sally},
  title   = {Quantum Causal Modelling},
  journal = {New Journal of Physics},
  volume  = {18},
  number  = {6},
  pages   = {063032},
  year    = {2016},
  doi     = {10.1088/1367-2630/18/6/063032}
}

@article{Leifer_2014,
  author  = {Leifer, Matthew S.},
  title   = {Is the Quantum State Real? An Extended Review of {$\psi$}-Ontology Theorems},
  journal = {Quanta},
  volume  = {3},
  number  = {1},
  pages   = {67--155},
  year    = {2014},
  doi     = {10.12743/quanta.v3i1.22},
  eprint  = {1409.1570},
  archivePrefix = {arXiv},
  primaryClass  = {quant-ph}
}

@article{Leifer_2017,
  author  = {Leifer, Matthew S. and Pusey, Matthew F.},
  title   = {Is a Time Symmetric Interpretation of Quantum Theory Possible Without Retrocausality?},
  journal = {Proceedings of the Royal Society A: Mathematical, Physical and Engineering Sciences},
  volume  = {473},
  number  = {2202},
  pages   = {20160607},
  year    = {2017},
  month   = jun,
  doi     = {10.1098/rspa.2016.0607}
}

@article{Conway_2006,
  author  = {Conway, John H. and Kochen, Simon},
  title   = {The Free Will Theorem},
  journal = {Foundations of Physics},
  volume  = {36},
  number  = {10},
  pages   = {1441--1473},
  year    = {2006},
  doi     = {10.1007/s10701-006-9068-6},
  eprint  = {quant-ph/0604079},
  archivePrefix = {arXiv}
}

@article{Goldstein_2010,
  author  = {Goldstein, Sheldon and Tausk, Daniel V. and Tumulka, Roderich and Zangh{\`i}, Nino},
  title   = {What Does the Free Will Theorem Actually Prove?},
  journal = {Notices of the American Mathematical Society},
  volume  = {57},
  number  = {11},
  pages   = {1451--1453},
  year    = {2010},
  eprint  = {0905.4641},
  archivePrefix = {arXiv},
  primaryClass  = {quant-ph}
}

@article{Vieira_2025,
  author  = {Vieira, Carlos and Ramanathan, Ravishankar and Cabello, Ad{\'a}n},
  title   = {Test of the Physical Significance of {Bell} Non-Locality},
  journal = {Nature Communications},
  volume  = {16},
  pages   = {4390},
  year    = {2025},
  month   = may,
  doi     = {10.1038/s41467-025-59247-7},
  issn    = {2041-1723}
}

@article{Tumulka_2009,
  author  = {Tumulka, Roderich},
  title   = {The Point Processes of the {GRW} Theory of Wave Function Collapse},
  journal = {Reviews in Mathematical Physics},
  volume  = {21},
  number  = {2},
  pages   = {155--227},
  year    = {2009},
  month   = mar,
  doi     = {10.1142/S0129055X09003608},
  eprint  = {0711.0035},
  archivePrefix = {arXiv},
  primaryClass  = {math-ph}
}

@book{Tumulka_2022,
  author    = {Tumulka, Roderich},
  title     = {Foundations of Quantum Mechanics},
  series    = {Lecture Notes in Physics},
  volume    = {1003},
  publisher = {Springer},
  address   = {Cham},
  year      = {2022},
  doi       = {10.1007/978-3-031-09548-1},
  isbn      = {978-3-031-09547-4}
}
\end{document}